\documentclass[12pt]{article}
\usepackage{a4}
\usepackage{latexsym}
\usepackage{amsmath}
\usepackage{amssymb}
\usepackage{tikz}
\usetikzlibrary{arrows.meta}
\usepackage{graphicx}

\usepackage{epic}
\newcommand{\nat}{\mathbb{N}}

\newcommand{\even}{{\sf even}}
\newcommand{\odd}{{\sf odd}}
\newcommand{\spir}{{\sf spir}}
\newcommand{\fib}{{\sf fib}}

\newtheorem{theorem}{Theorem}
\newtheorem{definition}{Definition}

\newenvironment{proof}{\noindent{\bf Proof:}}{$\Box$

\vspace{3 mm}}

\newcounter{opnr}
\newcommand{\examp}{

\vspace{2mm}

\addtocounter{opnr}{1} \noindent {\bf Example \theopnr.
}}

\title{Equality of morphic sequences}

\author{Hans Zantema\\
Department of Computer Science, TU Eindhoven\\
 P.O. Box 513,
5600 MB Eindhoven, the Netherlands \and
Radboud University Nijmegen\\
 P.O. Box 9010, 6500 GL Nijmegen, the Netherlands\\
{\tt h.zantema@tue.nl}}

\begin{document}

\maketitle

\begin{abstract}
Morphic sequences form a natural class of infinite sequences, typically defined as the coding of a fixed point of a morphism.
Different morphisms and codings may yield the same morphic sequence. This paper investigates how to prove that two such representations of a morphic sequence by morphisms represent the same sequence. In particular, we focus on the smallest representations of the subsequences of the binary Fibonacci sequence obtained by only taking the even or odd elements. The proofs we give are induction proofs of several properties simultaneously, and are typically found fully automatically by a tool that we developed.
\end{abstract}

\section{Introduction}
The simplest class of infinite sequences over a finite alphabet are {\em periodic}: sequences of the shape $u^\infty$ for some finite non-empty word $u$.
The one-but-simplest are {\em ultimately periodic}: sequences of the shape $v u^\infty$ for some finite non-empty words $u,v$. But these are still boring in some sense. More interesting are well-structured sequences that are simple to define, but are not ultimately periodic.
One class of such sequences are {\em morphic sequences}, being the topic of this paper. As a basic example consider the morphism $f$ replacing 0 by the word $01$ and replacing 1 by the symbol 0. Then we obtain the following sequence of words:
\[ f(0) = 01,\]
\[ f^2(0) = f(f(0)) = f(01) = 010,\]
\[ f^3(0) = f(f^2(0)) = f(010) = 01001,\]
\[ f^4(0) = f(f^3(0)) = f(01001) = 01001010,\]
and so on. We observe that in this sequence of words for every word its predecessor is a prefix, so we can take the limit of the sequence of words, being the binary {\em Fibonacci sequence} $\fib$, appearing as {\tt https://oeis.org/A003849} in the online encyclopedia of integer sequences. This is a typical example of a (pure) morphic sequence. In general, a {\em pure morphic sequence} over a finite alphabet $\Gamma$ is of the shape $f^\infty(a)$ for a finite alphabet $\Gamma$ and
 $f : \Gamma \to \Gamma^+$, $a \in \Gamma$ and $f(a) = au$, $u \in \Gamma^+$. Here $f^\infty(a)$ is defined as the limit of $f^n(a)$, which is well-defined as for every $n$ the word $f^n(a) = a u f(u) \cdots f^{n-1}(u)$ is a prefix of $f^{n+1}(a) = a u f(u) \cdots f^{n}(u)$. So
\[ f^\infty(a) \; = \; a u f(u) f^2(u) f^3(u) \cdots.\]
If we have a (typically smaller) finite alphabet $\Sigma$, and a {\em coding} $\tau : \Gamma \to \Sigma$, then a {\em morphic sequence} over the alphabet $\Sigma$ is defined to be of the shape $\tau(\sigma)$. Clearly any pure morphic sequence is morphic by choosing $\Sigma = \Gamma$ and $\tau$ to be the identity. Some more background on morphic sequences is given in Section \ref{secmorph}. Basic properties of morphic sequences are extensively described in the seminal book \cite{AS03}.

The class of morphic sequences is closed under several operations, like applying morphisms on them. Another operation under which morphisms are closed is $\even$, taking the even elements of a sequence, so
\[ \even(\sigma) \; = \; \sigma(0) \sigma(2) \sigma(4) \sigma(6) \cdots ,\]
formally defined by $\even(\sigma)(i) = \sigma(2i)$ for all $i \geq 0$. It follows from Theorem 7.9.1 in \cite{AS03} that $\even$ applied on a morphic sequence is morphic again. This theorem is quite deep and does not give a simple construction. In particular, the sequence $\even(\fib)$, being A339824 in the OEIS, is morphic. So it is a natural question to find a representation of $\even(\fib)$ of the shape $\tau(f^\infty(a))$ for a morphism $f$ as simple as possible. This may be done in several ways. Some ways are correct by construction. Others apply a brute force approach: let a computer program search for a representation $\tau(f^\infty(a))$ for which the alphabet of $f$ has $n$ elements, the length of $f(a)$ is $\leq k$ for all symbols $a$, and the first $N$ elements of $\even(\fib)$ and $\tau(f^\infty(a))$ coincide for some big number $N$. For $n = 5$ and $k = 2$ this gave rise to a representation for which we conjectured in the extended version  of \cite{Z24a} that indeed this represents $\even(\fib)$. Trying to solve this conjecture was the starting point of the current research. We wanted to prove that these different representations proposed for $\even(\fib)$ indeed represent the same sequence. By the approach presented in this paper this succeeded, even fully automatically, and also similarly for $\odd(\fib)$, being A339825 in the OEIS. In Section \ref{secfib}, all details will be presented.

However, the approach applies much more general than only for this particular example. In general, when trying to prove
$\tau(f^\infty(0)) = \rho(g^\infty(0))$ (without loss of generality we may write 0 for the start symbol for both $f$ and $g$) for morphisms $f$ and $g$ and codings $\tau$ and $\rho$, then first we scale up $f$ and $g$ until they have the same {\em dominant eigenvalue}, to be defined in Section \ref{secprel} and explained in Section \ref{secscu}. Then we search for prefixes $u$ of $f^\infty(0)$ and $v$ of $g^\infty(0)$
such that $\tau(f^n(u)) = \rho(g^n(v))$ for all $n$, to be proved by induction on $n$. From this property the goal $\tau(f^\infty(0)) = \rho(g^\infty(0))$ is easily concluded. The induction proof needs some {\em induction loading}: in fact for some $k$ we find $u_0,\ldots,u_{k-1}$,
$v_0,\ldots,v_{k-1}$ for which $u = u_0$ and $v = v_0$ and for which we prove by induction on $n$ that $\tau(f^n(u_i)) = \rho(g^n(v_i))$
for $i = 0,\ldots,k-1$. Surprisingly, we find a general pattern for which these $u_0,\ldots,u_{k-1}$, $v_0,\ldots,v_{k-1}$ may be easily found automatically, and also the induction proof is easily found fully automatically in many cases. We give a typical example.

\examp
\label{ex1}
We prove that $\fib$ (as defined before) equals $\rho(g^\infty(0))$ for $g, \rho$ defined by
$g(0) = 02$, $g(1) = 021$, $g(2) = 102$, $\rho(0) = \rho(1) = 0$, $\rho(2) = 1$. It turns out that here scaling up $f$ is required, that is, $f$ is replaced by $f^2$, by which $\fib = f^\infty(0)$ is not changed. By entering an input file only containing the original definition of $f$
($f(0) = 01, f(1) = 1$), $\tau(0) = 0$, $\tau(1) = 1$, and the just given definitions of $g, \rho$, then our tool immediately produces the following proof in LaTeX format:

\vspace{3mm}

\noindent Replace $f$ by $f^2$:
$f(0) = 010, f(1) = 01$.

\noindent Claim to be proved: $\tau(f^\infty(0)) = \rho(g^\infty(0))$.

\noindent We will prove the following 2 properties simultaneously by induction on $n$.

(0) $\tau(f^n(01)) = \rho(g^n(02))$.

(1) $\tau(f^n(0)) = \rho(g^n(1))$.

\noindent Then our claim follows from (0).

\noindent Basis $n=0$ of induction:

$\tau(f^0(01)) = 01 = \rho(g^0(02))$.

$\tau(f^0(0)) = 0 = \rho(g^0(1))$.

\noindent Basis of induction proved.

\noindent Induction step part (0):

$\tau(f^{n+1}(01)) = \tau(f^n(f(01))) = \tau(f^n(01001)) = $

$\tau(f^n(01)) \tau(f^n(0)) \tau(f^n(01)) =$   (by induction hypothesis)

$\rho(g^n(02)) \rho(g^n(1)) \rho(g^n(02)) = $

$\rho(g^n(02102)) = \rho(g^n(g(02))) = \rho(g^{n+1}(02)).$

\noindent Induction step part (1):

$\tau(f^{n+1}(0)) = \tau(f^n(f(0))) = \tau(f^n(010)) = $

$\tau(f^n(01)) \tau(f^n(0)) =$   (by induction hypothesis)

$\rho(g^n(02)) \rho(g^n(1)) = $

$\rho(g^n(021)) = \rho(g^n(g(1))) = \rho(g^{n+1}(1)).$

\noindent Induction step proved, hence claim proved.


\vspace{3mm}

This ends the generated proof. Note that this proof is elementary and very easy to check: it does not use anything else than the given definitions of $f, \tau, g, \rho$ and the principle of induction. This also holds for larger generated proofs. Then they typically contain quite some book keeping, but as being generated by a computer program, without the risk of making errors as in human generated proofs.

This paper is organized as follows. In Section \ref{secmorph} we give some general background of morphic sequences. In Section \ref{secprel} we give some notations and preliminaries. In Section \ref{secmain} we present the main theorem, and a special case that is easier to apply. In Section \ref{secscu} we describe scaling up: replacing $f$ and/or $g$ by some power in order to obtain the same dominant eigenvalue. In Section \ref{sectool} we describe our tool applying our main theorem, if necessary preceded by scaling up.
In Section \ref{seclim} we describe some limitations of the approach. In Section \ref{secfib} we present the proofs for $\even(\fib)$ and
$\odd(\fib)$ that were the starting point of this research. We conclude in Section \ref{secconcl}.

\section{Morphic sequences}
\label{secmorph}

In this section we give some more background on morphic sequences.

A typical example of a morphic sequence that is not pure morphic is
\[ \spir \; = \; 1101001000100001\cdots ,\]
consisting of infinitely many ones and for which the numbers of zeros in between two successive ones is $0,1,2,3,\ldots$, respectively.
Choosing $\Gamma = \{0,1,2\}$ and $f(2) = 21, f(1) = 01, f(0) = 0$, we obtain $f^2(2) = 2101$, $f^3(2) = 2101001$ and so on, yielding
\[ f^\infty(2) \; = \; 2101001000100001\cdots ,\]
for which indeed by choosing $\Sigma = \{0,1\}$, $\tau(0) = 0, \tau(1)= \tau(2) = 1$ we obtain $\spir = \tau(f^\infty(2))$, showing that
$\spir$ is morphic. But $\spir$ is not pure morphic. If it was then we have $\spir = f^\infty(1)$ for $f$ satisfying $f(1) = 1u$, $f(0) = v$, yielding contradictions for all cases: $u$ should contain a $0$, but no $1$ (otherwise $f^\infty(1)$ would contain a pattern $10^k1$ infinitely often for a fixed $k$), and $v$ should contain a $1$ (otherwise $f^\infty(1)$ would contain only a single $1$), but as $f^\infty(1)$ purely consists of $f(0) = v$ and $f(1) = 1u$ it cannot contain unbounded groups of zeros.

If you see this definition of morphic sequences for the first time, it looks quite ad hoc. However, there are several reasons to consider this class of morphic sequences as a natural class of sequences. One of them is the fact that there are several equivalent characterizations, see \cite{Z24a}. Another one is the observation that the class of morphic sequences is closed under several kinds of operations, as mentioned in the introduction.  A general result of this shape states that if applying a {\em finite state transducer} on a morphic sequence yields an infinite sequence, then it is always morphic too. Results like these are not easy, and can be found in \cite{AS03}, Corollary 7.7.5 and Theorem 7.9.1. An interesting hierarchy of (morphic) sequences is presented in \cite{EHK11}.

Another more esthetic argument to consider morphic sequences is that they give rise to interesting {\em turtle graphics}, \cite{Z16}.
Having a sequence $\sigma$ over a finite alphabet $\Sigma$, one chooses an angle $\phi(b)$ for every $b \in \Sigma$. Next a picture is drawn in the following way: choose an actual angle that is initialized in some way, and next proceed as follows: for $i = 0,1,2,3,\ldots$ the actual angle is turned by $\phi(\sigma(i))$ and after every turn a unit segment is drawn in the direction of the actual angle. For every segment its end point is used a the starting point for the next segment. In this way infinitely many segments are drawn. The resulting figure is called a {\em turtle figure}.

In the research paper \cite{Z16} and the book \cite{Z24} written for a wider audience, turtle figures are investigated for several morphic sequences. Sometimes this yields finite pictures, that is, after a finite but typically very big number of steps only segments will be drawn that have been drawn before by which the picture will not change any more.
In other cases the turtle figure will be {\em fractal}.

\begin{center}
\includegraphics[width=120mm]{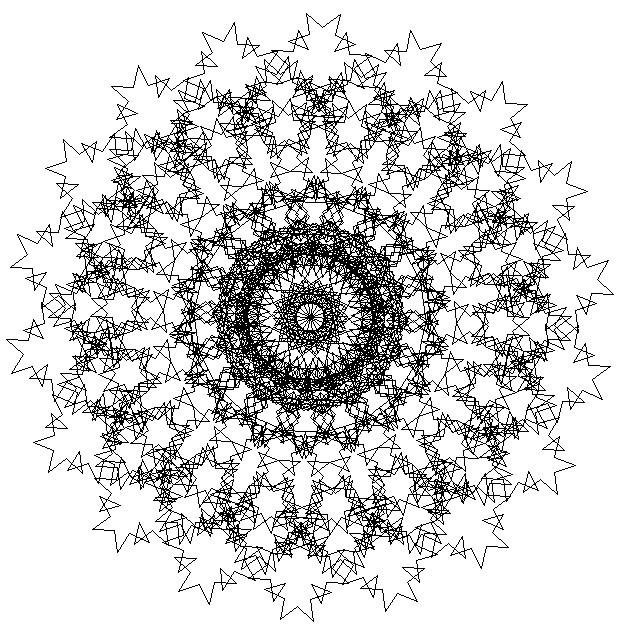}
\end{center}

To give the flavor of these figures, we just showed the finite turtle figure for the very simple pure morphic (and even 2-automatic) sequence $f^\infty(0) = 010001010100\cdots$ for $f(0) = 01$, $f(1) = 00$, sometimes called the {\em period doubling sequence}, and angles $\phi(0) = 140^o$ and $\phi(1) = -80^o$, so for every 0 the angle turns 140 degrees $= \frac{7 \pi}{9}$ to the left and for every 1 the angle turns 80 degrees $= \frac{4 \pi}{9}$ to the right, while after every turn a unit segment is drawn.

A simple program of only a few lines that draws this particular turtle figure shows that after 6000 steps not yet the full figure is drawn, but in \cite{Z24} it is shown that after $9216 = 2^{10}\times 3^2$ steps, only segments will be drawn that were drawn before, so the full picture will be equal to the picture drawn by only this finite part, and that is the picture we show here. Other morphic sequences (many of which not being automatic) and other angles give rise to a wide range of remarkable turtle figures, many of which are shown in \cite{Z16,Z24}. Turtle figures can be made for all kinds of sequences, but morphic sequences show up a nice balance between the boring and very regular figures for ultimately periodic sequences on the one hand, and the complete chaos of random sequences on the other hand. The notation for the sequence $\spir$ we saw is motivated by the fact that it gives rise to spiral shaped turtle figures.

Equational specifications of morphic sequences play a crucial role in earlier work, in particular \cite{Z09,ZR10} on investigating
when a specification has a unique solution, and in \cite{ZE11} focussing on proving equality automatically. Where \cite{ZE11} applies for a much more general setting than the current paper, it definitely fails for the cases presented in this paper.

\section{Notations and preliminaries}
\label{secprel}

For a {\em word} $w \in \Sigma^*$ over an alphabet $\Sigma$ we write $|w|$ for the length of $w$. The set of non-empty words over $\Sigma$ is denoted by $\Sigma^+$. The empty word of length 0 is denoted by $\epsilon$.

We write $\nat$ for the set of natural numbers, including $0$.

A {\em sequence} $\sigma$ over an alphabet $\Sigma$ can be seen as a mapping $\sigma : \nat \to \Sigma$, so
\[ \sigma \; = \; \sigma(0) \sigma(1) \sigma(2) \cdots .\]
So by sequence we always mean an infinite sequence. In some texts they are called {\em streams}.
The set of sequences over $\Sigma$ is denoted by $\Sigma^\infty$.

Morphisms $f : \Sigma \to \Sigma^+$ can also be applied on words and sequences, for instance, $f(\sigma)$ is the concatenation of the following words:
\[ f(\sigma) \; = \; f(\sigma(0))f(\sigma(1))f(\sigma(2)) \cdots.\]

From the introduction we recall:
\begin{definition}
A {\em pure morphic sequence} over a finite alphabet $\Gamma$ is of the shape $f^\infty(a)$ for a finite alphabet $\Gamma$ and
 $f : \Gamma \to \Gamma^+$, $a \in \Gamma$ and $f(a) = au$, $u \in \Gamma^+$.

A {\em morphic sequence} over a finite alphabet $\Sigma$ is of the shape $\tau(\sigma)$ for some coding $\tau : \Gamma \to \Sigma$ and some pure morphic sequence $\sigma$ over some finite alphabet $\Gamma$.
\end{definition}

Here $f^\infty(a)$ is defined as the limit of $f^n(a)$, which is well-defined as for every $n$ the word $f^n(a) = a u f(u) \cdots f^{n-1}(u)$ is a prefix of $f^{n+1}(a) = a u f(u) \cdots f^{n}(u)$. So
\[ f^\infty(a) \; = \; a u f(u) f^2(u) f^3(u) \cdots.\]
It is easy to see that it is the unique {\em fixed point} of $f$ starting in $a$, that is, $f(f^\infty(a)) = f^\infty(a)$.
Elements not occurring in $f^\infty(a)$ may be ignored and removed, so without loss of generality we will assume that every element of $\Gamma$ occurs in $f^\infty(a)$.

Without loss of generality we may (and often will) assume that $\Gamma = \{ 0,1,\ldots,n-1\}$ and $a = 0$.

In the above notation the coding $\tau : \Gamma \to \Sigma$ is lifted to $\tau : \Gamma^\infty \to \Sigma^\infty$ in which $\tau$ is applied on all separate elements.

A slightly more relaxed definition of morphic sequence also allows $f(b) = \epsilon$ for some elements $b \in \Gamma$. In \cite{AS03}, Theorem 7.5.1, it is proved that this is equivalent to our definition, so we may and will always assume that $f(b) \in \Gamma^+$ for all $b \in \Gamma$.

For integers $0 \leq k \leq m$ for any sequence $\sigma$ over $\Sigma$ we write $\sigma_{k,m} = \sigma(k) \sigma(k+1) \cdots \sigma(m-1) \in \Sigma^+$, being a word of length $m-k$, so is empty if $m=k$. A word of the shape $\sigma_{0,m}$ is called a {\em prefix} of $\sigma$.

To a morphism $f : \Gamma \to \Gamma^+$ for $\Gamma =  \{ 0,1,\ldots,n-1\}$ an $n \times n$ matrix $F$ over the natural numbers is associated in which for every $i,j$ the number $F_{ij}$ is the number of occurrences of $i$ in the word $f(j)$. It is easily checked that for every $k \geq 1$ the matrix associated to $f^k$ is $F^k$. Under mild conditions for such a matrix over natural numbers there is always a positive real eigenvalue such that the absolute value of every other eigenvalue is less or equal to this particular real one.
This is called the {\em dominant eigenvalue}, see e.g., \cite{D98,R14a,R14b}, and coincides with $\lim_{k \to \infty} \frac{|f^{k+1}(u)|}{|f^k(u)|}$ for a suitable initial word $u$. As we assumed that every element of $\Gamma$ occurs in $f^\infty(0)$, any word $u$ containing $0$ will do, in particular $u = 0$. The approach to approximate the dominant eigenvalue by $\frac{|f^{k+1}(u)|}{|f^k(u)|}$ for some large enough number $k$ is often called the {\em power method}.

A morphism $f$ and its associated matrix $F$ are called {\em primitive} if for some $k > 0$ all entries of $F^k$ are positive, that is, for every $i,j$ the symbols $j$ occurs in $f^k(i)$. For instance, the function $f$ presented for defining $\fib$ is primitive, and the function $f$ presented for defining $\spir$ is not primitive.

The following theorem is a reformulation of the main result of \cite{D98}.

\begin{theorem}
\label{thmd}
If $f,g$ are primitive and $\tau(f^\infty(0)) = \rho(g^\infty(0))$ is not periodic, then there exist $k,m$ such that $f^k$ and $g^m$ have the same dominant eigenvalue.
\end{theorem}

For primitive morphisms this can be seen as a generalization of Cobham's theorem, \cite{C69}, also being Theorem 11.2.1 in \cite{AS03}.
This restricts to the case where both $f$ and $g$ are {\em uniform}. Here a morphism $f$ is called uniform if it is $k$-uniform for some $k$, and it is called $k$-uniform if $|f(a)| = k$ for all symbols $a$. In that case the dominant eigenvalue of $f$ is $k$.

\section{The main theorem}
\label{secmain}

Without loss of generality any finite alphabet $\Gamma$ is written as $\Gamma = {0,\ldots,n-1}$ for $n = |\Gamma|$.

Two arbitrary morphic sequences over any alphabet $\Sigma$ may be written as $\tau(f^\infty(0))$ and $\rho(g^\infty(0))$, for
$\Gamma_f  = {0,\ldots,n_f-1}$, $\Gamma_g  = {0,\ldots,n_g-1}$, $f : \Gamma_f \to \Gamma_f^+$, $g : \Gamma_g \to \Gamma_g^+$, $\tau : \Gamma_f \to \Sigma$, $\rho : \Gamma_g \to \Sigma$, where both $f(0)$ and $g(0)$ start in 0 and have length $> 1$.

The theorem presented in this section gives a criterion that is easily checked in any particular case, and from which $\tau(f^\infty(0)) = \rho(g^\infty(0))$ can be concluded.

\begin{theorem}
\label{thmmain}
In the above setting, let $u_i,v_i$ be non-empty words over $\{0,1,\ldots,n_f-1\}$, $\{0,1,\ldots,n_g-1\}$, respectively, for $i = 0, 1, \ldots, n-1$ for some number $n > 0$. For $i = 0, 1, \ldots, n-1$ let $w_i$ be some word over $\{0, 1, \ldots, n-1\}$, write $w_i = w_{i,0} w_{i,1} \cdots w_{i,k_i-1}$ where $k_i = |w_i|$.  Assume that for $i = 0, 1, \ldots, n-1$ the following holds for $v_i = a_0 \cdots a_{k-1}$:
\begin{itemize}
\item $\tau(u_i) = \rho(v_i)$,
\item $f(u_i) = u_{w_{i,0}} u_{w_{i,1}} \cdots u_{w_{i,k_i-1}}$, and
\item $g(v_i) = v_{w_{i,0}} v_{w_{i,1}} \cdots v_{w_{i,k_i-1}}$.
\end{itemize}
Moreover, assume that both $u_0$ and $v_0$ start in 0.

Then $\tau(f^\infty(0)) = \rho(g^\infty(0))$.
\end{theorem}

\begin{proof}
We prove the following property by induction on $m$:

For every $m \geq 0$ and for every $i = 0, 1, \ldots, n-1$ we have $\tau(f^m(u_i)) = \rho(g^m(v_i))$.

The basis of the induction, for $m=0$, reads  $\tau(u_i) = \rho(v_i)$ for $i = 0, 1, \ldots, n-1$, which holds by assumption.

For the induction step we assume that for every $i = 0, 1, \ldots, n-1$ we have $\tau(f^m(u_i)) = \rho(g^m(v_i))$ for some $m$.
For $i = 0, 1, \ldots, n-1$ we observe that
\[ \begin{array}{rcll}
\tau(f^{m+1}(u_i)) &=& \tau(f^m(f(u_i))) & \\
&=& \tau(f^m(u_{w_{i,0}} u_{w_{i,1}} \cdots u_{w_{i,k_i-1}})) & \mbox{ (by assumption)} \\
&=& \rho(g^m(v_{w_{i,0}} v_{w_{i,1}} \cdots v_{w_{i,k_i-1}})) & \mbox{(by induction hypothesis, $k_i \times$)} \\
&=& \rho(g^m(g(v_i))) & \mbox{(by assumption)} \\
&=& \rho(g^{m+1}(v_i)), & \end{array}\]
concluding the induction proof.

One of the claims we have proved by induction is $\tau(f^m(u_0)) = \rho(g^m(v_0))$, for every $m \geq 0$. As $u_0$ and $v_0$ start in $0$,
for every $m \geq 0$ we have that both $\tau(f^m(0))$ and $\rho(g^m(0))$ are prefixes of  $\tau(f^m(u_0)) = \rho(g^m(v_0))$, from which we conclude  $\tau(f^\infty(0)) = \rho(g^\infty(0))$.
\end{proof}

As in Example \ref{ex1} given in the introduction choose $f(0)= 010$, $f(1)=01$, $\tau(0) = 0$, $\tau(1) = 1$, $g(0)= 02$, $g(1)=021$, $g(2) = 102$, $\rho(0) = \rho(1) = 0$, $\rho(2) = 1$. Then Theorem \ref{thmmain} applies by choosing $u_0 = 01$, $u_1 = 0$, $v_0 = 02$, $v_1 = 1$, for which the first condition $\tau(u_i) = \rho(v_i)$ is easily checked for $i=0,1$. Moreover, $f(u_0) = 01001 = u_0 u_1 u_0$,
$g(v_0) = 02102 = v_0 v_1 v_0$, $f(u_1) = 010 = u_0 u_1$ and $g(v_1) = 021 = v_0 v_1$, so by choosing $w_0 = 010$ and $w_1 = 01$, by which the remaining conditions of Theorem \ref{thmmain} have been checked, and equality $\tau(f^\infty(0)) = \rho(g^\infty(0))$ may be concluded. In fact the proof given in Example \ref{ex1} is exactly the proof of Theorem \ref{thmmain} instantiated for this particular instance.

The main challenge to apply Theorem \ref{thmmain} is to find suitable words $u_i, v_i, w_i$ for which the conditions hold. The next theorem elaborates the particular case of choosing $n = n_g$ and $v_i = w_i = g(i)$ for all $i = 0, 1, \ldots, n-1$. For this choice the last condition $g(v_i) = v_{w_{i,0}} v_{w_{i,1}} \cdots v_{w_{i,k_i-1}}$ is obtained for free. It remains to choose $u_i$ such that $\tau(u_i) = \rho(v_i)$ and $f(u_i) = u_{w_{i,0}} u_{w_{i,1}} \cdots u_{w_{i,k_i-1}}$.
Assume that $\tau(f^\infty(0)) = \rho(g^\infty(0))$ is expected to hold after checking that the first $N$ elements of $\tau(f^\infty(0))$ and $\rho(g^\infty(0))$ coincide for some big $N$. Then a natural choice is to choose $u_i$ to be the factor of $\tau(f^\infty(0))$ on the position of the first occurrence of $g(i)$ in $\rho(g^\infty(0))$. Applied to this setting, Theorem \ref{thmmain} is reformulated to the following theorem.

\begin{theorem}
\label{thmmain2}
For $i = 0, 1, \ldots, n_g-1$ let $w_i$ be the prefix in front of the first occurrence of $i$ in $g^\infty(0)$, so
$g^\infty(0)(|w_i|) = i$ and $i$ does not occur in $g^\infty(0)_{0,|w_i|}$, and write $u_i = f^\infty(0)_{|g(w_i)|, |g(w_i i)|}$.

For $i = 0, 1, \ldots, n_g-1$ assume that $\tau(u_i) = \rho(g(i))$ and $f(u_i) = u_{a_0} \cdots u_{a_{k-1}}$ for $g(i) = a_0 \cdots a_{k-1}$.

Then $\tau(f^\infty(0)) = \rho(g^\infty(0))$.
\end{theorem}

%
%
%
%
Note that for applying Theorem \ref{thmmain2} for given $f,\tau,g,\rho$ all conditions can be checked directly, not depending on any further choice.

\examp Choose $f(0) = 01, f(1) = 101$, $\tau$ is the identity, hence may be omitted, and
$g(0) = 021, g(1) = 102, g(2) = 02$, $\rho(0) = 0$ and $\rho(1) = \rho(2) = 1$. It turns out that $f^\infty(0) = \rho(g^\infty(0))$. In the setting of Theorem \ref{thmmain2} we obtain $u_0 = 011$, $u_1 = 101$ and $u_2 = 01$. We check that $u_0 = \rho(g(0))$, $u_1 = \rho(g(1))$ and $u_2 = \rho(g(2))$, and $f(u_0) = 01101101 = u_0 u_2 u_1$, $f(u_1) = 10101101 = u_1 u_0 u_2$ and $f(u_2) = 01101 = u_0 u_2$, so by Theorem \ref{thmmain2} we obtain $f^\infty(0) = \rho(g^\infty(0))$.

\vspace{2mm}

So the simplest approach to apply Theorem \ref{thmmain} is via checking the requirements of its special instance Theorem \ref{thmmain2}.
This is only some book keeping, easily to be done by a computer program, and that is indeed what we did. Although this works well for quite some examples, for some examples Theorem \ref{thmmain2} shows up not to apply, while a
proof is possible by choosing words $u_i$, $v_i$, $w_i$ satisfying all conditions of Theorem \ref{thmmain}. A simple instance is the following.

\examp
Let $f,g$ be defined by $f(0) = 02, f(1) = 101, f(2) = 10$ and
$g(0) = 0210, g(1) = 1, g(2) = 10$. Let $\tau, \rho$ both being the identity, so they may be omitted. Then
$f^\infty(0) = g^\infty(0)$. However, Theorem \ref{thmmain2} fails to prove this, also if $f,g$ are swapped. But Theorem \ref{thmmain}
succeeds by choosing $n=2$, $u_0 = v_0 = 021$, $u_1 = v_1 = 01$, as $f(u_0) = 0210101 = u_0 u_1 u_1$, $g(v_0) = 0210101 = v_0 v_1 v_1$, $f(u_1) = 02101 = u_0 u_1$ and $g(v_1) = 02101 = v_0 v_1$, hence proving $f^\infty(0) = g^\infty(0)$. This proof is found by our tool.

\vspace{2mm}

In cases
where Theorem \ref{thmmain2} applies, often other instances of Theorem \ref{thmmain} apply that give simpler proofs. For instance, for Example \ref{ex1}  Theorem \ref{thmmain2} applies by proving the three properties $\tau(f^n(01)) = \rho(g^{n+1}(0))$, $\tau(f^n(010)) = \rho(g^{n+1}(1))$ and $\tau(f^n(001)) = \rho(g^{n+1}(2))$ simultaneously by induction on $n$, where the proof given in Example \ref{ex1} was simpler, only involving two properties to be proved by induction simultaneously. Our final tool applies a more sophisticated search for suitable words $u_i$, $v_i$, $w_i$ in order to apply Theorem \ref{thmmain}. This will be discussed in Section \ref{sectool}. Before this approach applies some scaling up may be required, which will be discussed now.

\section{Scaling up}
\label{secscu}
When proving $\tau(f^\infty(0)) = \rho(g^\infty(0))$ by Theorem \ref{thmmain} or its corollary Theorem \ref{thmmain2}, we find words $u_0,v_0$ starting in 0 such that for every $m \geq 0$ we have $\tau(f^m(u_0)) = \rho(g^m(v_0))$. Hence $f^m(u_0)$ and $g^m(v_0)$ have the same length for every $m \geq 0$. In particular, $\lim_{m \to \infty} \frac{|f^{m+1}(u_0)|}{^|f^m(u_0)|} = \lim_{m \to \infty} \frac{|g^{m+1}(v_0)|}{^|g^m(v_0)|}$. As both $u_0$ and $v_0$ contain the symbol 0, these limits coincide with dominant eigenvalues of $f$ and $g$. Hence Theorems \ref{thmmain} and \ref{thmmain2} only may be successful if $f$ and $g$ have the same dominant eigenvalue.

However, if $f$ and $g$ do not have the same dominant eigenvalue, it is often possible to scale up by choosing two numbers $k,m$ such that $f^k$ and $g^m$ have the same dominant eigenvalue. In fact, Theorem \ref{thmd} states that for $f,g$ being primitive and the result sequence being non-periodic, this is always possible. Note that by replacing $f$ by $f^k$ and $g$ by $g^m$, the sequences $\tau(f^\infty(0))$ and $\rho(g^\infty(0))$ do not change. For instance, $\fib = f^\infty(0)$ for $f(0) = 01, f(1) = 0$ also satisfies
$\fib = h^\infty(0)$ for $h = f^3$ defined by $h(0) = 01001, h(1) = 010$. The dominant eigenvalue of $f$ is $\phi = \frac{1+\sqrt{5}}{2}$, the dominant eigenvalue of $h = f^3$ is $\phi^3$.

In fact a very first of example of this scaling up we already saw in Example \ref{ex1}, where $f$ was replaced by $f^2$ in order to obtain the same dominant eigenvalues for $f$ and $g$. Next we give an example in which both $f$ and $g$ have to be scaled in order to obtain the same dominant eigenvalues.

\examp

Two other representations of $\fib$ are $\tau(f^\infty(0))$ and $\rho(g^\infty(0))$ for $\tau,\rho,f,g$ defined by
\[ \tau(0) = \tau(1) = 0, \; \tau(2) = 1, \; \rho = \tau,\]
\[f(0) = 0210, \; f(1) = 02102, \; f(2) = 2021,\]
\[g(0) = 02, \; g(1) = 021, \; g(2) = 102.\]
Here the dominant eigenvalue of $f$ is $\phi^3$ and the dominant eigenvalue of $g$ is $\phi^2$. So the only hope to apply
Theorem \ref{thmmain} or Theorem \ref{thmmain2} is by replacing $f$ by $f^2$ and $g$ by $g^3$, by which both get dominant eigenvalue
$\phi^6$. Keeping the same names $f,g$ we then have
\[ \begin{array}{rcl}
f(0) &=& 02102021021020210, \\
f(1) &=& 021020210210202102021, \\
f(2) &=& 20210210202102102, \\
g(0) &=& 0210202102102, \\
g(1) &=& 021020210210202102021, \\
g(2) &=& 021020210210202102102. \end{array}\]
Now indeed Theorem \ref{thmmain} applies by choosing $u_0 = v_0 = 02$ and $u_1 = v_1 = 1$, as is found by our tool, hence proving simultaneously
$\tau(f^n(02)) = \rho(g^n(02))$ and $\tau(f^n(1)) = \rho(g^n(1))$ for all $n$.

Two shorter separate proofs of $\tau(f^\infty(0)) = \fib$ and $\tau(g^\infty(0)) = \fib$ are possible, but the above proof may be given fully automatically without knowing that both  $\tau(f^\infty(0))$ and $\tau(g^\infty(0))$ are equal to $\fib$.

\section{The tool}
\label{sectool}
We developed a tool for proving equality of morphic sequences fully automatically, based on Theorem \ref{thmmain}. As input the definitions of $f$, $\tau$, $g$ and $\rho$ are entered in order to prove $\tau(f^\infty(0)) = \tau(g^\infty(0))$.

First the dominant eigenvalues of $f$ and $g$ are approximated by means of the {\em  power method}: the dominant eigenvalue of $f$ is approximated by $\frac{|f^{n+1}(0)|}{|f^n(0)|}$ for some number $n$ big enough, and similarly for $g$. In our tool we choose $n = 8$. If these approximations are close to each other, then the approach of Theorem \ref{thmmain} is tried directly, otherwise $f$ and/or $g$ are replaced by its square or cube until the approximated dominant eigenvalues are close. If this fails, then the attempt is given up. Of course also trying higher powers than only squares or cubes would be possible too, but as this will give rise to very long strings $f^n(a)$ or $g^n(a)$, we do not do this.

The next challenge is to choose $u_i, v_i, w_i$ for which the conditions of Theorem \ref{thmmain} hold. For every $u_i, v_i$ the property
$\tau(f^n(u_i) = \rho(g^n(v_i))$ should hold for every $n$. So every pair $(u_i,v_i)$ should be {\em safe}, where a pair $(u,v)$ of words is called safe if $|u| = |v|$ and $|f(u)| = |g(v)|$. We start by choosing $u_0$ and $v_0$ to be the smallest non-empty prefixes of $f^\infty(0)$ and $g^\infty(0)$, respectively, such that $(u_0,v_0)$ is safe. This is done by checking whether the prefixes of length $1,2,3,\ldots$ give rise to safe pair until a safe pair is found. If this is not found for a length $\leq 10$, then the attempt is given up.

After this safe pair $(u_0,v_0)$ has been found, a next requirement is $f(u_0) = u_{w_{0,0}} u_{w_{0,1}} \cdots u_{w_{0,k_0-1}}$, and
$g(v_0) = v_{w_{0,0}} v_{w_{0,1}} \cdots v_{w_{0,k_0-1}}$. We choose $w_{0,0} =0$: $f(u_0)$ is a prefix of $f^\infty(0)$ and $g(v_0)$ is a prefix of $g^\infty(0)$ of the same length. So $(u_{w_{0,1}},v_{w_{0,1}})$ should be a safe pair, where $u_{w_{0,1}}$ starts in $f(u_0)$ just after $u_{w_{0,0}}$ and $v_{w_{0,1}}$ starts in $g(v_0)$ just after $v_{w_{0,0}}$. We search for the smallest safe pair starting on this position. If it $(u_0,v_0)$ again, we choose $w_{0,1} = 0$, otherwise we define it to be $(u_1,v_1)$ and choose $w_{0,1} = 1$. This process continues: after this position we search in $f(u_0)$ and $g(v_0)$ for the next safe pair. If this coincides with $(u_i,v_i)$ that was already defined, then this one is chosen, otherwise a new $(u_i,v_i)$ is defined. After $f(u_0)$ and $g(v_0)$ have been split up into safe pairs $(u_i,v_i)$ in this way, fully defining $w_0$ and meeting the conditions $f(u_0) = u_{w_{0,0}} u_{w_{0,1}} \cdots u_{w_{0,k_0-1}}$ and $g(v_0) = v_{w_{0,0}} v_{w_{0,1}} \cdots v_{w_{0,k_0-1}}$, the process continues in the same way for $f(u_1)$ and $g(v_1)$, starting at the front. And next the same is done for all remaining $f(u_i)$ and $g(v_i)$ for which $(u_i,v_i)$ have been defined. If in this way at some point no new pairs $(u_i,v_i)$ are created, and all these pairs $(u_i,v_i)$ give rise to
$f(u_i) = u_{w_{i,0}} u_{w_{i,1}} \cdots u_{w_{i,k_i-1}}$ and $g(v_i) = v_{w_{i,0}} v_{w_{i,1}} \cdots v_{w_{i,k_i-1}}$, then all conditions of Theorem \ref{thmmain} are met as long as $\tau(u_i) = \rho(v_i)$ for all $i$, which is checked by the tool too.

If this all succeeds, then the desired equality $\tau(f^\infty(0)) = \rho(g^\infty(0))$ is concluded by Theorem \ref{thmmain}. However,
the output of the tool does not refer to Theorem \ref{thmmain}, but instantiates the proof of Theorem \ref{thmmain} by the words $u_i,v_i,w_i$ found by the above construction. In this way the generated proof only makes use of the given definition and the principle of induction. In fact Example \ref{ex1} as given in the introduction shows the literal output of the tool.

The tool was written in Freepascal that is freely available as being open source, and may be run on several platforms. The input consists of a direct encoding of $f, \tau, g, \rho$. More precisely, the first line consists of $n_f$, being the size of the signature of $f$. The next $n_f$ lines contain the $n_f$ words $f(0), f(1),\ldots, f(n_f-1)$. The next line encodes $\tau$ by the word $\tau(0) \tau(1) \cdots \tau(n_f-1)$ of length $n_f$. Next exactly the same is done for $g, \rho$. So the input for Example \ref{ex1} reads

\begin{verbatim}
2
01
0
01
3
02
021
102
001
\end{verbatim}

The source code of the tool to prove equality of morphic sequences is stored in the file {\tt eq.txt}. Apart from this standard tool also two variants are available. One has source code {\tt eql.txt} and does exactly the same but produces its result in LaTeX format.
Applying this variant on the above input exactly gives the proof as presented in Example \ref{ex1}. The other variant is a basic version that after some possible upscaling only tries to apply Theorem \ref{thmmain2}. The source code of this basic version is available as {\tt eqb.txt}. For many examples both the standard version and the basic version give a correct proof, but often the proof given by the standard version is simpler.

After compiling (by {\tt fpc eq.txt}) the executable file {\tt eq} is available. If a file {\tt inp.txt} starts by the input in the format as given above, then by
calling {\tt eq < inp.txt} our tool is applied on the corresponding input, and produces the resulting proof if possible. For the variants
{\tt eql.txt} and {\tt eqb.txt} this works similarly.

For over 20 particular examples, including Examples 1 to 6 in this paper (named {\tt ex1.txt} until {\tt ex6.txt}) and several representations for $\even(\fib)$ and $\odd(\fib)$, we provide a file containing:
\begin{itemize}
\item the example in the input format as given above, by which the tool may be applied directly on this file,
\item the output of the standard tool applied on the example if it is successful, and
\item the output of the basic tool applied on the example if it is successful.
\end{itemize}

The source code of the programs, their Windows executables and all example files are available at:

{\tt https://github.com/hzantema/Equality-of-morphic-sequences}

\section{Limitations}
\label{seclim}

In order to prove $\tau(f^\infty(0) = \rho(g^\infty(0)$ for any morphic sequences $\tau(f^\infty(0)$ and $\rho(g^\infty(0)$ our approach is first scale up $f$ and $g$ until they have the same dominant eigenvalue, and then apply Theorem \ref{thmmain}. In this section we give a few examples showing why this approach may fail for different reasons.

A first reason may be that upscaling until $f$ and $g$ have the same dominant eigenvalue is not possible. According to Theorem \ref{thmd} this may only occur if either $f$ or $g$ is not primitive or the sequence $\tau(f^\infty(0) = \rho(g^\infty(0)$ is periodic.

\examp
Let $f(0) = 01, f(1) = 11, g(0) = 01, g(1) = 111$ and let $\tau = \rho$ be the identity. Then $f^\infty(0) = 0 1^\infty = g^\infty(0)$.
However, the dominant eigenvalues of $f$ and $g$ are 2 and 3, respectively, for which upscaling until the domimant eigenvalues are equal is not possible. So our approach will fail. Note that neither $f$ nor $g$ is primitive.

\vspace{2mm}

We are not aware of examples that are not ultimately periodic and non-primitive for which upscaling fails. But the next example is
not ultimately periodic and non-primitive and we give an argument why Theorem \ref{thmmain} does not apply.

\examp
Let $f(0) = 01, f(1) = 21, f(2) = 2, g(0) = 012, g(1) = 12, g(2) = 2$ and let $\tau = \rho$ be the identity. Then
\[f^\infty(0) = 0 1212212^312^41\cdots = g^\infty(0).\]
Both $f$ and $g$ have dominant eigenvalue 1.
However, Theorem \ref{thmmain} does not apply, since if it does then for every $n \geq 0$ we obtain $f^n(u_0) = g^n(v_0)$ for strings $u_0,v_0$ starting in 0. For $n=0$ this implies $u_0 = v_0$. Since $|f(0)| < |g(0)|, \; |f(1)| = |g(1)|, \; |f(2)| = |g(2)|$ and $u_0 = v_0$ contains a symbol 0, this implies  $|f(u_0)| < |g(v_0)|$, contradiction.

\section{Subsequences of $\fib$}
\label{secfib}

As mentioned in the introduction, the starting point of this research was finding a small representation for
\[\even(\fib) \; = \; \fib(0) \fib(2) \fib(4) \fib(6) \cdots \; = \]
\[ 00110011000100010001100110001000100\cdots \]
As $\fib = \tau(f^\infty(0))$ for $|f(a)| \leq 2$ for all symbols $a$, we wondered whether such a representation is also possible for $\even(\fib)$, for a probably larger alphabet. For $n=2,3,4,\ldots$ we applied brute force computer search to find $\tau,f$ on an alphabet of $n$ elements for which all strings $f(a)$ have length 1 or 2 and for which the first 40 elements of $\tau(f^\infty(0))$ coincide with the first 40 elements of $\even(\fib)$. For $n=2,3,4$ this gave no solution, but for $n=5$ the following solution was found:
\[ f_e(0) = 01, f_e(1) = 2, f_e(2) = 31, f_e(3) = 04, f_e(4) = 0,\]
\[ \tau_e(0) = \tau_e(1) = 0, \tau_e(2) = \tau_e(3) = \tau_e(4) = 1.\]
It was easily checked that not only $\tau_e(f_e^\infty(0))$ and $\even(\fib)$ coincide at the first 40 positions, but also at the first one million positions. So this strongly suggests that $\tau_e(f_e^\infty(0)) = \even(\fib)$, but the challenge now was to prove this, preferably elementary. This was formulated as a challenge both in \cite{BD24} and in the extended version of \cite{Z24a}. Jeffrey Shallit informed us that this could be proved by his tool Walnut. Nevertheless, we still were looking for a general approach that applies for arbitrary morphic sequences.

A crucial insight was the following encoding of $\even(\fib)$ as a morphic sequence proposed by Henk Don that is correct by construction.
Generalized to arbitrary arithmetic subsequences of $\fib$ this is described in detail in \cite{BD24}. Restricted to $\even(\fib)$ and
$\odd(\fib)$ defined by
\[\odd(\fib) \; = \; \fib(1) \fib(3) \fib(5) \fib(7) \cdots \]
the encoding is as follows. By upscaling a factor 3, we obtain $\fib = f^\infty(0)$ for $f$ defined by $f(0) = 01001, f(1) = 010$. The key property is that now both $f(0)$ and $f(1)$ have odd length, by which $f(u)$ has always even length for $u$ of length 2. In particular:
\[ f(01) \; = \; 01 \; 00 \; 10 \; 10, \]
\[ f(00) \; = \; 01 \; 00 \; 10 \; 10 \; 01,\]
\[ f(10) \; = \; 01 \; 00 \; 10 \; 01. \]
So by decomposing $\fib$ into words of length 2, and renaming $01$ to $a$, $00$ to $b$ and $10$ to $c$, the sequence $\fib$ is also obtained as $g^\infty(a)$ for
\[ g(a) = abcc, \; g(b) = abcca, \; g(c) = abca.\]
The sequence $\even(\fib)$ is obtained by replacing every word of length 2 in this decomposition by its first symbol, so $\rho(a) = \rho(b) = 0, \rho(c) = 1$.
So by construction we have $\even(\fib) = \rho(g^\infty(a))$, and similarly $\odd(\fib) = \tau(g^\infty(a))$ for $\tau(a) = 1$, $\tau(b) = \tau(c) = 0$.

As a consequence, for proving $\tau_e(f_e^\infty(0)) = \even(\fib)$ it suffices to prove $\tau_e(f_e^\infty(0)) = \rho(g^\infty(a))$, and for this the approach of this paper applies. In fact first we found a proof based on Theorem \ref{thmmain2}, but later on this was simplified to the following proof that is fully automatically produced by our tool (in which $a,b,c$ are renamed to $0,1,2$ and the subscript $e$ is omitted):

\[f(0) = 01,
f(1) = 2,
f(2) = 31,
f(3) = 04,
f(4) = 0,\]
\[\tau(0) = 0,
\tau(1) = 0,
\tau(2) = 1,
\tau(3) = 1,
\tau(4) = 1,\]
\[g(0) = 0122,
g(1) = 01220,
g(2) = 0120,
\rho(0) = 0,
\rho(1) = 0,
\rho(2) = 1.\]
Replace $f$ by $f^3$:
\[f(0) = 01231,
f(1) = 042,
f(2) = 01031,
f(3) = 01201,
f(4) = 012.\]
Claim to be proved: $\tau(f^\infty(0)) = \rho(g^\infty(0))$.

\noindent We will prove the following 5 properties simultaneously by induction on $n$.

$(0) \; \tau(f^n(012)) = \rho(g^n(012)).$

$(1) \; \tau(f^n(31)) = \rho(g^n(20)).$

$(2) \; \tau(f^n(0)) = \rho(g^n(1)).$

$(3) \; \tau(f^n(42)) = \rho(g^n(22)).$

$(4) \; \tau(f^n(01)) = \rho(g^n(00)).$

\noindent Then our claim follows from $(0)$.

\noindent Basis n=0 of induction:
\[\tau(f^0(012)) = 001 = \rho(g^0(012)),\;\;
\tau(f^0(31)) = 10 = \rho(g^0(20)),\]
\[\tau(f^0(0)) = 0 = \rho(g^0(1)), \;\;
\tau(f^0(42)) = 11 = \rho(g^0(22)),\]
\[\tau(f^0(01)) = 00 = \rho(g^0(00)).\]
\noindent Basis of induction proved.

\noindent Induction step part (0):
\[\tau(f^{n+1}(012)) = \tau(f^n(f(012))) = \tau(f^n(0123104201031)) = \]
\[\tau(f^n(012)) \tau(f^n(31)) \tau(f^n(0)) \tau(f^n(42)) \tau(f^n(01)) \tau(f^n(0)) \tau(f^n(31)) =  \]
\[ \mbox{ (by induction hypothesis)}\]
\[\rho(g^n(012)) \rho(g^n(20)) \rho(g^n(1)) \rho(g^n(22)) \rho(g^n(00)) \rho(g^n(1)) \rho(g^n(20)) = \]
\[\rho(g^n(0122012200120)) =
\rho(g^n(g(012))) =
\rho(g^{n+1}(012)).\]
Induction step part (1):
\[\tau(f^{n+1}(31)) =
\tau(f^n(f(31))) =
\tau(f^n(01201042)) =\]
\[\tau(f^n(012)) \tau(f^n(01)) \tau(f^n(0)) \tau(f^n(42)) =  \; \mbox{ (by induction hypothesis)}\]
\[\rho(g^n(012)) \rho(g^n(00)) \rho(g^n(1)) \rho(g^n(22)) =\]
\[\rho(g^n(01200122)) =
\rho(g^n(g(20))) =
\rho(g^{n+1}(20)).\]
Induction step part (2):
\[\tau(f^{n+1}(0)) =
\tau(f^n(f(0))) =
\tau(f^n(01231)) =
\tau(f^n(012)) \tau(f^n(31)) =  \]
\[ \mbox{ (by induction hypothesis)}\]
\[\rho(g^n(012)) \rho(g^n(20)) =
\rho(g^n(01220)) =
\rho(g^n(g(1))) =
\rho(g^{n+1}(1)).\]
Induction step part (3):
\[\tau(f^{n+1}(42)) =
\tau(f^n(f(42))) =
\tau(f^n(01201031)) =\]
\[\tau(f^n(012)) \tau(f^n(01)) \tau(f^n(0)) \tau(f^n(31)) =  \; \mbox{ (by induction hypothesis)}\]
\[\rho(g^n(012)) \rho(g^n(00)) \rho(g^n(1)) \rho(g^n(20)) =\]
\[\rho(g^n(01200120)) =
\rho(g^n(g(22))) =
\rho(g^{n+1}(22)).\]
Induction step part (4):
\[\tau(f^{n+1}(01)) =
\tau(f^n(f(01))) =
\tau(f^n(01231042)) = \]
\[\tau(f^n(012)) \tau(f^n(31)) \tau(f^n(0)) \tau(f^n(42)) =  \; \mbox{ (by induction hypothesis)}\]
\[\rho(g^n(012)) \rho(g^n(20)) \rho(g^n(1)) \rho(g^n(22)) =\]
\[\rho(g^n(01220122)) =
\rho(g^n(g(00))) =
\rho(g^{n+1}(00)).\]
Induction step proved, hence claim proved.

\vspace{3mm}

This concludes the generated proof, hence proving that indeed $\tau_e(f_e^\infty(0)) = \even(\fib)$. A similar proof 
for a slightly larger encoding of both $\even(\fib)$ and $\odd(\fib)$ was given in \cite{AS23}.

It is a natural question whether
$\tau_e(f_e^\infty(0))$ is the smallest possible representation of $\even(\fib)$. For making this question more precise, we have to precisely define what is meant by {\em smallest}. As in \cite{BD24} we define {\em smallest} to be of minimal complexity, where the complexity of a representation $\tau(f^\infty)$ of a morphic sequence is defined to be the sum of $|f(a)|$ where $a$ runs over the alphabet. So the complexity of the standard representation of $\fib$ is 3, and the complexity of the representation $\tau_e(f_e^\infty(0))$ of $\even(\fib)$ is 8. Note that this complexity does not coincide with the number of states in the corresponding automaton as described in \cite{Z24a,R14b}, but with the number of transitions $=$ arrows.

In \cite{BD24} it was conjectured that the minimal complexity of $\even(\fib)$ is 8, realized by $\tau_e,f_e$, and the minimal complexity of $\odd(\fib)$ is 10. Later on, Wieb Bosma found a representation of $\odd(\fib)$ of complexity 9, namely $f_o, \tau_o$ defined by
\[ f_o(0) = 01, f_o(1) = 51, f_o(2) = 30, f_o(3) = 4, f_o(4) = 3, f_o(5) = 2,\]
\[\tau_o(0) = \tau_o(2) = \tau_o(4) = 11 \;  \tau_o(1) = \tau_o(3) = \tau_o(5) = 0.\]
This was found by brute force search, and it was checked that the first several million elements of $\odd(\fib)$ and $\tau_o(f_o^\infty(0))$ coincide. Proving that $\odd(\fib) = \tau_o(f_o^\infty(0))$ was left as a challenge for our approach, and indeed it succeeds. As observed earlier, $\odd(\fib) = \rho(g^\infty(0))$ for $g, \rho$ defined by
\[g(0) = 0122,
g(1) = 01220,
g(2) = 0120,
\rho(0) = 1,
\rho(1) = 0,
\rho(2) = 0.\]
So we have to prove that $\tau_o(f_o^\infty(0)) = \rho(g^\infty(0))$, and after omitting the subscripts our tool immediately finds a proof, similar to our proof of $\tau_e(f_e^\infty(0)) = \even(\fib)$. As a first step again $f$ is replaced by $f^3$, now yielding
$f(0) = 0151251, f(1) = 30251, f(2) = 30151, f(3) = 4, f(4) = 3, f(5) = 401$. Next for proving $\tau(f^\infty(0)) = \rho(g^\infty(0))$
the following 5 properties are proved simultaneously by induction on $n$:

(0) $\tau(f^n(01512513)) = \rho(g^n(01220122))$.

(1) $\tau(f^n(02514)) = \rho(g^n(00120))$.

(2) $\tau(f^n(013)) = \rho(g^n(012))$.

(3) $\tau(f^n(02513)) = \rho(g^n(00122))$.

(4) $\tau(f^n(01514)) = \rho(g^n(01220))$.

As the rest is only basic bookkeeping as in our proof of $\tau_e(f_e^\infty(0)) = \even(\fib)$, further details are omitted here, they may be found as {\tt oddfib.txt} in the repository

{\tt https://github.com/hzantema/Equality-of-morphic-sequences}.

A more thorough brute force analysis by Wieb Bosma showed that $f_o,\tau_o$ is indeed the smallest representation of $\odd(\fib)$ of complexity 9, and it is the only one of complexity 9. He also showed that $f_e,\tau_e$ is indeed the smallest representation of $\even(\fib)$ of complexity 8, but this is not the only one. There are exactly two of complexity 8, and the other one is $\tau(f^\infty(0))$ for $\tau, f$ defined by
\[ f(0) = 01, f(1) = 2, f(2) = 34, f(3) = 0, f(4) = 32,\]
\[ \tau(0) =  \tau(1) =  \tau(4) = 0,  \tau(2) =  \tau(3) = 1.\]
Correctness of this latter representation is proved in a similar way by our tool, two proofs are found in the repository.
The repository also contains proofs
{\tt oddfib2.txt} and {\tt oddfib3.txt} for other representations of $\odd(\fib)$, and proofs {\tt evfib.txt}, {\tt evfib2.txt}, {\tt evfib3.txt} and {\tt evfib4.txt} for $\even(\fib)$, the first one coinciding with the proof we gave and the last two dealing with the above mentioned alternative representation of complexity 8.

\section{Concluding remarks}
\label{secconcl}

The class of morphic sequences is closed under several operations, like applying morphisms or taking arithmetic subsequences. So applying such operations on some simple morphic sequence yields a morphic sequence again, but often it is unclear how to represent it efficiently in the standard way by a morphism and a coding. The starting point of this paper was to investigate this for some very simple instances: apply $\even$ or $\odd$ on the Fibonacci sequence $\fib$. Finding some representation that is correct for the first $N$ elements for some fixed number $N$ is feasible by applying brute force computing. Checking that a resulting representation is correct for the first $N$ elements for some much bigger number $N$ is also feasible by a simple computer program. But then it remains to prove that it is correct for all elements rather than only an initial part. That may be solved as soon as we have a technique to formally prove that two given distinct representations represent the same morphic sequence. This is the main topic of this paper. The main theorem is Theorem \ref{thmmain}, by which equality is concluded if some words may be chosen satisfying a range of technical conditions. A tool was developed that tries to choose the desired words in a systematic way. It turned out that in this way proving correctness of simple representations for $\even(\fib)$ and $\odd(\fib)$ could be done fully automatically. Apart from these particular examples we gave several more examples for which our approach applies fully automatically, all available in a public repository.

Apart from the evidence of these examples and a few arguments why our approach fails in some particular cases given in Section \ref{seclim}, we have a quite limited understanding of the power of our approach. Even the very basic question of decidability of the problem to establish whether two representations give the same morphic sequence seems to be open for the general case. This problem is also called HD0L $\omega$-equivalence, and has been solved for primitive morphisms in \cite{D12}.

A next challenge would be to apply the same approach to find small representations for other arithmetic subsequences of $\fib$, like the sequence obtained by only keeping the elements on positions divisible by 3, or other operations on $\fib$ like applying a morphism. And all this could be done for other morphic sequences than only for $\fib$. Apart from finding morphic sequences of low complexity that are correct for the first $N$ elements by brute force computing, then a main challenge is to find reasonable small representations that are correct by construction.

A main issue is finding a smallest representation for a particular sequence, Here smallest refers to lowest complexity. A similar complexity analysis for operations on the more restricted class of automatic sequences was given in \cite{ZB21}.

{\bf Acknowledgement:} We want to thank Jean-Paul Allouche, Wieb Bosma, Henk Don, Fabien Durand and Jeffrey Shallit for fruitful suggestions and discussions.

\bibliography{reftrs}

\end{document}